\documentclass[aps,twocolumn,pra]{revtex4-1}
\usepackage{graphics}
\usepackage{dcolumn}
\usepackage{bm}
\usepackage{amsmath}
\usepackage{hyperref}
\hypersetup{colorlinks=true, urlcolor=blue}
\usepackage{color}
\usepackage{tikz}
\usepackage{multirow}

\usepackage{float}
\usepackage[caption = false]{subfig}
\usepackage{newlfont}
\usepackage{amssymb}
\usepackage{amsfonts}
\usepackage{amsthm}
\usepackage{graphicx}
\usepackage{bm}

\newtheorem{theorem}{Theorem}
\newtheorem{corollary}{Corollary}[theorem]
\newtheorem{remark}{Remark}[theorem]

\def \be{\begin{equation}}
\def \ee{ \end{equation} }

\begin{document}

\definecolor{red}{rgb}{1,0,0}
\title{Conditions for monogamy of quantum correlations in multipartite systems}
\author{Asutosh Kumar}
\email{asukumar@hri.res.in}

\affiliation{Harish-Chandra Research Institute, Chhatnag Road, Jhunsi, Allahabad 211 019, India}

\begin{abstract}
Monogamy of quantum correlations is a vibrant area of research because of its potential applications in several areas in quantum information ranging from quantum cryptography to co-operative phenomena in many-body physics.
In this paper, 
we investigate conditions under which monogamy is preserved for functions of quantum correlation measures.
We prove that a monogamous measure remains monogamous on raising its power, and a non-monogamous measure remains non-monogamous on lowering its power.  
We also prove that monogamy of a convex quantum correlation measure for arbitrary multipartite pure quantum state leads to its monogamy for mixed states in the same Hilbert space. Monogamy of squared negativity for mixed states and that of entanglement of formation follow as corollaries of our results. 
\end{abstract}

\maketitle

\section{Introduction}

Quantum correlations \cite{horodecki09,modidiscord}, of both entanglement \cite{horodecki09} and information-theoretic \cite{modidiscord} paradigms, is an indispensable resource in quantum information theory \cite{nielsen}.
While entanglement measures capture the nonseparability of two or more subsystems, information-theoretic measures like quantum discord \cite {hv, oz} can detect nonclassical properties even in separable states.
It is desirable that a quantum correlation measure ${\cal Q}$ belonging to either of two above classes satisfies certain basic properties \cite{horodecki09,modidiscord, adesso-general} such as \emph{positivity}, ${\cal Q}(\rho_{AB}) \geq 0$, and \emph{monotonicity}, i.e., is non-increasing under a suitable set of local quantum operations and classical communications
[in particular, invariance under local unitaries $U_A \otimes V_B$, ${\cal Q}(\rho_{AB}) = {\cal Q}(U_A \otimes V_B \rho_{AB} U_A^{\dagger} \otimes V_B^{\dagger})$, as well as \emph{no-increase upon attaching a local pure ancilla}, ${\cal Q}(\rho_{AB}) \geq {\cal Q}(\rho_{AB} \otimes |0\rangle_C \langle 0|)$].
These properties are valid for several known measures of quantum correlations, including all entanglement
measures. In particular, positivity and invariance under local unitaries are standard requirements \cite{std-req}.

Quantum correlations, entanglement in particular, is crucial in quantum information processing and quantum computation \cite{nielsen}, in describing area laws \cite{area-law1, area-law2, area-law3, area-law4, area-law5, area-law6, area-law7, area-law8, area-law9, area-law10, area-law11, area-law12}, in quantum phase
transition and detecting other cooperative quantum phenomena in various interacting quantum many-body systems \cite{many-body1, many-body2, many-body3, many-body4}. 
Hence, quantum correlations form a fundamental aspect of modern physics and a key enabler in quantum communication and computation technologies.
Being a resource, quantification of quantum correlations is important.
Although a number of correlation measures for bipartite (qubit) systems have been studied extensively in last few decades, there has not been much investigation of multipartite correlations owing to difficulty in defining multipartite correlations.

The concept of monogamy \cite{monogamy, osborne} is a distinguishing feature of quantum correlations,
which sets it apart from classical correlations.
Monogamy of quantum correlations is an active area of research, and has found potential applications in quantum information theory like in quantum key distribution
\cite{QKD-terhal, secureQKD, cryptoreview}, in classifying quantum states \cite{dvc, giorgi, discord-hri}, in distinguishing orthonormal quantum bases \cite{asu-bell}, in black-hole physics \cite{mono-blackhole1, mono-blackhole2}, to study frustated spin systems \cite{frust-spin}, etc. 
Morever, it has proved to be a useful tool in exploring multipartite nonclassical correlations \cite{monogamy, osborne, mono-dis2}. 
Qualitatively, monogamy of quantum correlations places certain restrictions on distribution of quantum correlations 
of one fixed party with other parties of a multipartite system. In particular, \emph{if party A in a tripartite
system ABC is maximally quantum correlated with party B, then A cannot be correlated at all
to the third party C}. 
This is true for all quantum correlation measures, and is a departure from classical correlations which are not bound to such constraints. That is, classical correlations do not satisfy a monogamy constraint \cite{no-mono-classical1, no-mono-classical2, no-mono-classical3, no-mono-classical4, no-mono-classical5, no-mono-classical6, no-mono-classical7, no-mono-classical8, no-mono-classical9, no-mono-classical10}.
In other words, monogamy forbids free sharing of quantum correlations
among the constituents of a multipartite quantum system. 
This is a nonclassical property in the sense that such constraints are not observed even in the maximally classically-correlated systems like
\begin{equation}
\rho_{ABC}=\frac 12 \big(|000\rangle \langle 000| + |111\rangle \langle 111| \big).
\label{eq:mcs}
\end{equation}

However, two or more parties in a multipartite quantum state do not necessarily always share maximal quantum correlation, and are thus able to share some correlations with other parties, although in a restrictive manner. 
Thus, monogamy relations help in determining entanglement structure in the multipartite setting. 
Furthermore, it has been argued to be a consequence of the no-cloning theorem \cite{mono-no-cloning-cons1, mono-no-cloning-cons2, mono-no-cloning-cons3}.
Monogamy, like entanglement \cite{schrodinger}, appears to be the trait of multipartite entangled quantum systems. 
Interestingly, the notion of monogamy is not restricted only to quantum correlation measures, but has spawned its wing 
in other quantum properties such as Bell inequality \cite{bell-ineq1, bell-ineq2, bell-ineq3}, quantum steering \cite{qsteering}, 
and contextual inequalities \cite{context-ineq1, context-ineq2, context-ineq3}.
A quantum correlation measure that satisfies the ``monogamy inequality'' for all quantum states is termed ``monogamous''. However, we know that not all quantum correlation measures, even for three-qubit states, satisfy monogamy. Entanglement measures
such as concurrence \cite{concurrence1, concurrence2}, entanglement of formation \cite{eof}, negativity \cite {negativity}, etc., apart from information-theoretic measures such as quantum discord \cite{hv, oz} are known to be, in general, non-monogamous. 
In recent developments on monogamy, we have seen that exponent of a quantum correlation measure and multipartite quantum states play a remarkable role in characterization of monogamy \cite{salini, asu-multi}. 
A non-monogamous quantum correlation measure can become monogamous, for three or more parties, when its power is increased \cite{salini}. For instance, concurrence, entanglement of formation, negativity, quantum discord are non-monogamous for three-qubit states, but their squared versions are monogamous. 
In particular, it has been shown that \emph{monotonically increasing functions of any quantum correlation can make all multiparty states monogamous} with respect to that measure \cite{salini}. 
We note that \emph{the increasing function of the correlation measure under consideration satisfies all the necessary properties for being a quantum correlation measure including posititvity and monotonicity under local operations}, mentioned above. Furthermore, the function can be so chosen that it is reversible \cite{rev-func1, rev-func2}, such that the information about quantum correlation in the state under consideration, after applying the function on the quantum correlation remains intact.
The power of a correlation measure is an example of such a function. It is interesting to note that the function $f(x)=x^{\alpha}$ is \emph{concave} for $0 < \alpha \leq 1$ and \emph{convex} for $1 \leq \alpha \leq \infty$ on the interval $(0, \infty)$. 
The power function has an intrinsic geometric interpretation. The power defines the slope of the graph.
The higher power, the graph is nearer to the vertical axis.
It has been found that several measures of quantum correlations like squared concurrence \cite{monogamy, osborne}, squared negativity \cite{mono-neg2-fan, mono-neg2-vidal, mono-neg2-kim}, squared quantum discord \cite{mono-dis2}, global quantum discord \cite{mono-gqd1, mono-gqd2}, squared entanglement of formation \cite{mono-eof2-bai, mono-eof2-fanchini}, Bell inequality \cite{mono-BI1, mono-BI2, mono-BI3}, EPR steering \cite{mono-epr-steering1, mono-epr-steering2}, contextual inequalities \cite{mono-contextual1, mono-contextual2}, etc. exhibit monogamy property.
Thus, we observe that the convexity plays a key role in establishing monogamy of quantum correlations.
In another case, non-monogamous quantum correlation measures become monogamous, for moderately large number of parties \cite{asu-multi}.

The motivation behind this paper is three-fold. In this letter, we have asked
(i) under what conditions monogamy property of quantum correlations is preserved, (ii) does monogamy for arbitrary pure multipartite state lead to monogamy of mixed states, and (iii) are there more general and stronger monogamy relations different from the standard one in Eq. (\ref{eq:std-mono-ineq}). 
We prove that while a monogamous measure remains monogamous on raising its power, a non-monogamous measure remains non-monogamous on lowering its power. 
We also prove that monogamy of a convex quantum correlation measure for an arbitrary multipartite pure quantum state leads to its monogamy for the mixed state in the same Hilbert space. Monogamy of squared negativity for mixed states and that of entanglement of formation follow as direct corollaries. 
Authors of Ref. \cite{hier-eof2} have proposed following two conjectures regarding monogamy of squared entanglement of formation in multiparty systems: \emph{the squared entanglement of formation may be monogamous for multipartite (i) $2 \otimes d_2 \otimes \cdots \otimes d_n$, and (ii) arbitrary $d$-dimensional, quantum systems}. Our previous result partially answers these conjectures in the sense that it now only remains to prove the monogamy of the squared entanglement of formation for pure states in arbitrary dimensions. 
We have further given hierarchical monogamy relations, and a strong monogamy inequality 
\begin{align}
{\cal Q}^{\alpha}(\rho_{AB}) \geq \frac{1}{2^{n-1}-1}\sum_{X} {\cal Q}^{\alpha}\left(\rho_{AX}\right) \geq \sum_{j} {\cal Q}^{\alpha}\left(\rho_{AB_j}\right),
\end{align}
where $X=\{B_{i_1},\cdots,B_{i_k}\}$ is a nonempty proper subset of $B \equiv \{B_1,B_2,\cdots,B_n \}$, and $\alpha \geq 1$ is some positive real number.

This letter is organized as follows. In Section \ref{monogamy}, we succintly review the notion of monogamy of quantum correlations. While the main results of this letter are presented in Section \ref{results}, we give a summary in Section \ref{sec:summary}.

\section{Monogamy of Quantum Correlations}
\label{monogamy}
Consider that ${\cal Q}$ is a bipartite entanglement measure. If for a multipartite quantum system described by a state $\rho_{AB_1 B_2\cdots B_n} \equiv \rho_{AB}$, the following inequality
\begin{equation}
{\cal Q}(\rho_{A(B_1\cdots B_n)}) \geq \sum_{j=1}^n{\cal Q}(\rho_{AB_j}),
\label{eq:std-mono-ineq}
\end{equation}
holds, then the state $\rho_{AB}$ is said to be monogamous under the quantum correlation measure ${\cal Q}$ \cite{monogamy, osborne}. Otherwise, it is non-monogamous. Moreover, the deficit between the two sides is referred to as 
monogamy score \cite{monogamyscore}, and is given by
\begin{equation}
\delta{\cal Q}={\cal Q}(\rho_{AB}) - \sum_{j=1}^n{\cal Q}(\rho_{AB_j}).
\label{eq:mono-score}
\end{equation}
Monogamy score can be interpreted as residual entanglement of the bi-partition $A:rest$ of an \(n\)-party state that cannot be accounted for by the entanglement of two-qubit reduced density matrices separately.

It should be noted here that the monogamy inequality in Eq. (\ref{eq:std-mono-ineq}) is just one constraint on the distribution of quantum correlations. Suppose ${\cal Q}$ does not obey the monogamy relation in 
Eq. (\ref{eq:std-mono-ineq}), then is it non-monogamous? Can it be shared freely among the constituent parties? 
It may happen that it obeys the following constraint
\begin{equation}
\label{eq:mono-ineq}
\sum_{j=1}^n{\cal Q}(\rho_{AB_j}) \leq b (\neq n), 
\end{equation} 
and be still monogamous. Here we assume that ${\cal Q}$ is normalized, i.e., $0 \leq {\cal Q} \leq 1$. Numerical evidence of such a limitation was observed for entanglement of formation and concurrence in Ref. \cite{mono-eof2-fanchini} for three-qubit systems.
 
Can there be more general and stronger monogamy relations than in Eq. (\ref{eq:std-mono-ineq})? Considerable attemps have been made to address this question from different perspectives \cite{hier-eof2, adesso-gaussian1, adesso-gaussian2, adesso-general, adesso-sm-4qb} recently. 


\section{Results}
\label{results}
In this section, we prove that a monogamous measure remains monogamous on raising its power, a non-monogamous measure remains non-monogamous on lowering its power, and monogamy of a convex quantum correlation measure for arbitrary multipartite pure quantum states leads to its monogamy for the mixed states. We also examine tighter monogamy inequalities compared to the standard one in Eq. (\ref{eq:std-mono-ineq}), and hierarchical monogamy relations. Throughout our discussion we denote the multipartite quantum state $\rho_{AB_1B_2\cdots B_n}$ by $\rho_{AB}$, unless stated otherwise.

\begin{theorem}
\label{thm-raising}
{\bf (Monogamy preserved for raising of power)} For an arbitrary multipartite quantum state $\rho_{AB}$, if ${\cal Q}^r(\rho_{AB}) \geq \sum_j {\cal Q}^r(\rho_{AB_j})$ then ${\cal Q}^{\alpha}(\rho_{AB}) \geq \sum_j {\cal Q}^{\alpha}(\rho_{AB_j})$ for $\alpha \geq r \geq 1$.
\end{theorem} 

\begin{proof}
We have the inequalities $(1+x)^t  \geq 1+x^t$ and 
$(\sum_i x_i^t)^{\frac{s}{t}} \geq \sum_i x_i^{s}$ where $0 \leq x,x_i \leq 1$ and $s \geq t \geq 1$. Now there exists $1 \leq k \leq n$ such that $\sum_{j \neq k} {\cal Q}^r(\rho_{AB_j}) \geq {\cal Q}^r(\rho_{AB_k})$. Now ${\cal Q}^r(\rho_{AB}) \geq \sum_j {\cal Q}^r(\rho_{AB_j})$ implies 
\begin{align}
{\cal Q}^{\alpha}(\rho_{AB})& \geq \left(\sum_j {\cal Q}^r(\rho_{AB_j})\right)^{\frac{\alpha}{r}} \nonumber \\
& = \left(\sum_{j \neq k} {\cal Q}^r(\rho_{AB_j})\right)^{\frac{\alpha}{r}}\left(1+\frac{{\cal Q}^r(\rho_{AB_k})}{\sum_{j \neq k} {\cal Q}^r(\rho_{AB_j})}\right)^{\frac{\alpha}{r}} \nonumber \\
& \geq \left(\sum_{j \neq k} {\cal Q}^r(\rho_{AB_j})\right)^{\frac{\alpha}{r}}\left(1+\left(\frac{{\cal Q}^r(\rho_{AB_k})}{\sum_{j \neq k} {\cal Q}^r(\rho_{AB_j})}\right)^{\frac{\alpha}{r}}\right) \nonumber \\
& = \left(\sum_{j \neq k} {\cal Q}^r(\rho_{AB_j})\right)^{\frac{\alpha}{r}}+{\cal Q}^{\alpha}(\rho_{AB_k}) \nonumber \\
& \geq \left(\sum_{j \neq k} {\cal Q}^{\alpha}(\rho_{AB_j})\right)+{\cal Q}^{\alpha}(\rho_{AB_k}) \nonumber \\
& = \sum_{j} {\cal Q}^{\alpha}(\rho_{AB_j}),
\end{align}
thus proving the theorem.
\end{proof}
The first and second inequalities respectively follow from the fact that $(1+x)^t  \geq 1+x^t$ and 
$(\sum_i x_i^t)^{\frac{s}{t}} \geq \sum_i x_i^{s}$ where $0 \leq x,x_i \leq 1$ and $s \geq t \geq 1$. 
This theorem can be viewed as an extension of the key result in Ref. \cite{salini} that a non-monogamous quantum correlation measure will become monogamous for some value when its power is raised.

\begin{theorem}
\label{thm-lowering}
{\bf (Non-monogamy preserved for lowering of power)} For an arbitrary multipartite quantum state $\rho_{AB}$, if ${\cal Q}^r(\rho_{AB}) \leq \sum_j {\cal Q}^r(\rho_{AB_j})$ then ${\cal Q}^{\alpha}(\rho_{AB}) \leq \sum_j {\cal Q}^{\alpha}(\rho_{AB_j})$ for $\alpha \leq r$.
\end{theorem}

\begin{proof}
The inequality ${\cal Q}^r(\rho_{AB}) \leq \sum_j {\cal Q}^r(\rho_{AB_j})$ implies ${\cal Q}^{\alpha}(\rho_{AB}) \leq \left(\sum_j {\cal Q}^r(\rho_{AB_j})\right)^{\frac{\alpha}{r}}$. Now above theorem can be proved by using the inequality $(1+x)^t  \leq 1+x^t$, for $x > 0$ and $t \leq 1$, repeatedly.
\end{proof}

\begin{remark}
Theorems \ref{thm-raising} and \ref{thm-lowering} ensure that varying the exponent preserves monogamy (non-monogamy) relations of monogamous (non-monogamous) correlation measures.
Recently, it was shown in Ref. \cite{ent-mono-qubit} that, for multiqubit systems, the $r^{th}$-power of concurrence is monogamous for $r \geq 2$ while non-monogamous for $r \leq 0$, and the $r^{th}$-power of entanglement of formation (EoF) is monogamous for $r \geq \sqrt{2}$. These observations are consistent with above theorems. 
Similarly, negativity, quantum discord for three-qubit pure states, contextual inequalities, etc., will remain monogamous for $r \geq 2$. Also, quantum work-deficit, for all three-qubit pure states, will remain monogamous for the fifth power and higher \cite{salini}.
\end{remark}

\begin{remark}
Note that, at first sight, it seems that Theorems \ref{thm-raising} and \ref{thm-lowering} are rather about properties of abstract functions that can describe not only nonclassical correlations but any other property. 
We wish to note however that the theorems are not true for an arbitrarily chosen physical property.
For instance, for the mixture $\rho_{ABC}=\frac 12 \big(|000\rangle \langle 000| + |111\rangle \langle 111| \big)$ in Eq. (\ref{eq:mcs}), the classical correlation, irrespective of the definition used, in all the bi-partitions A:(BC), A:B, and A:C is unity, after a suitable normalization. In this case, raising the power to any value, however large, of the classical correlation, won’t make it monogamous. This example illustrates the fact why raising or lowering of powers to nonclassical correlations is important and necessary. Thus, the above two theorems should be seen mainly in the context of nonclassical correlations.
\end{remark}

\begin{figure}[htb]
\includegraphics[width=3.5cm]{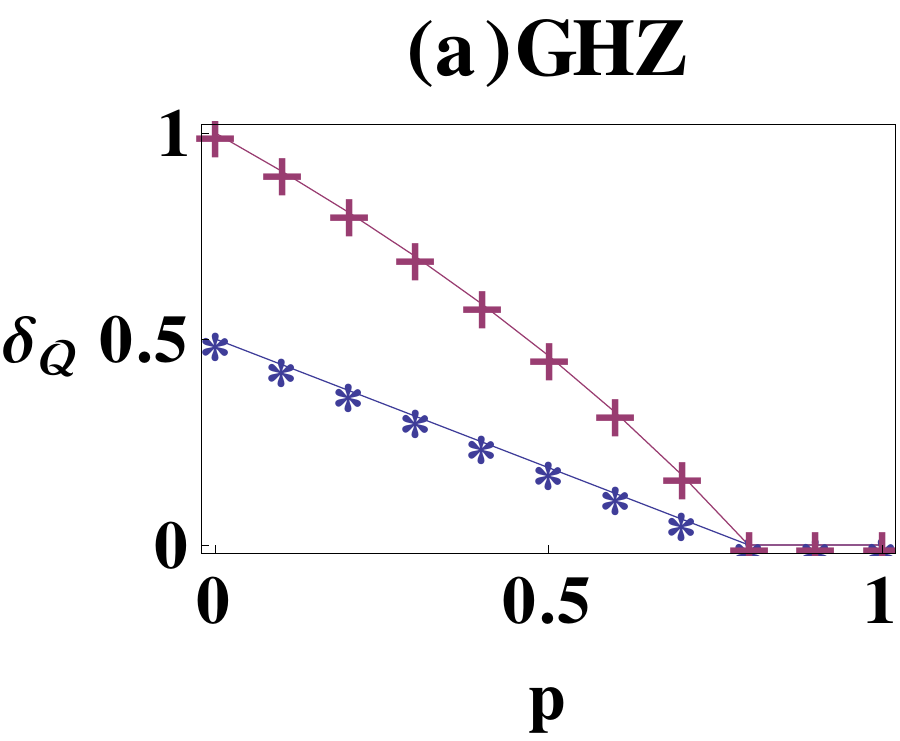} \hspace{0.5cm}
\includegraphics[width=3.5cm]{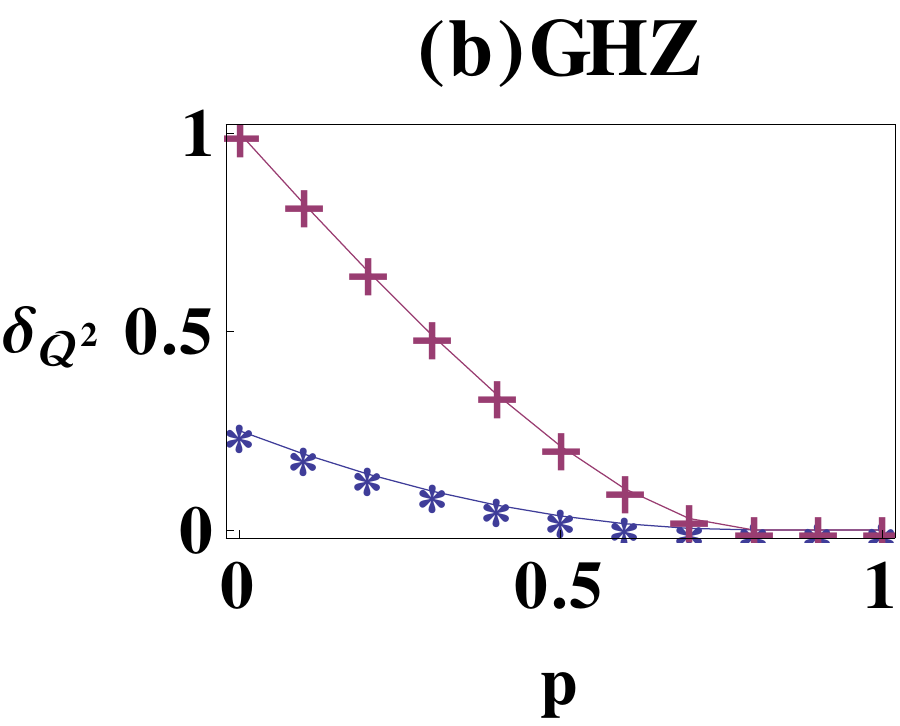}
\caption{(Color online.) Plots of monogamy scores, $\delta_{{\cal Q}^r}(\rho_{ABC})={\cal Q}^r(\rho_{A(BC)})-{\cal Q}^r(\rho_{AB})-{\cal Q}^r(\rho_{AC})$, of negativity (stars) and logarithmic-negativity (pluses) against noise parameter $p$ of GHZ state mixed with white noise, $\rho_{ABC}=(1-p)|GHZ\rangle \langle GHZ| + p \frac{\mathbb{I}}{8}$. For illustration, power of the quantum correlation measures we have considered is (a) $r=1$, and (b) $r=2$. Here we see that, for GHZ state, both negativity and logarithmic-negativity are monogamous for $r \geq 1$. While $x$-axis is dimensionless, $y$-axis is in ebits.}
\label{fig:ghz}
\end{figure}
\begin{figure}[htb]
\includegraphics[width=3.5cm]{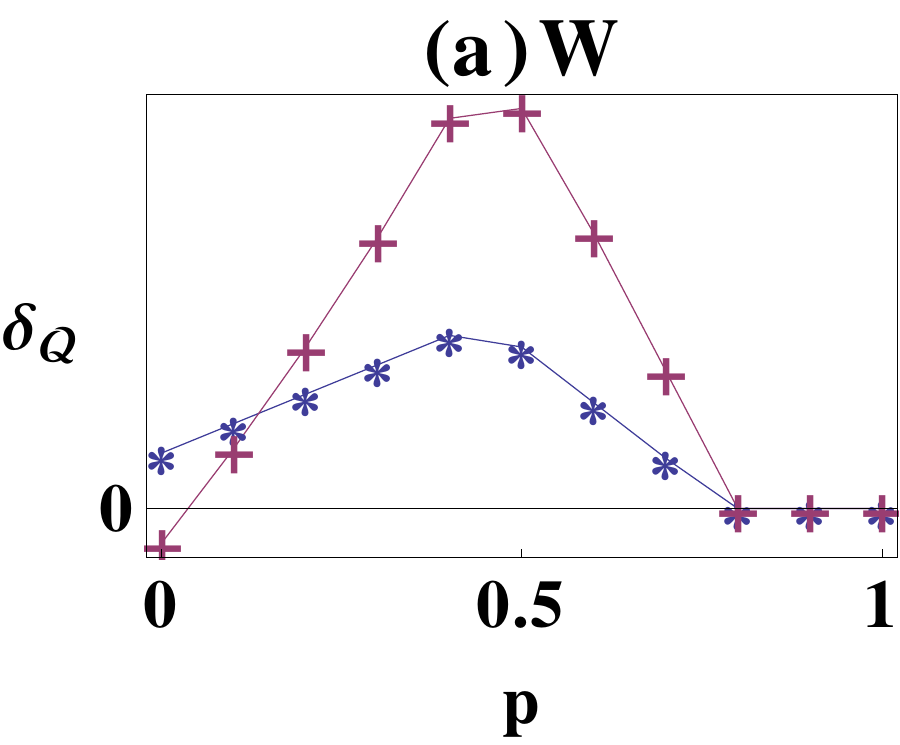} \hspace{0.5cm}
\includegraphics[width=3.5cm]{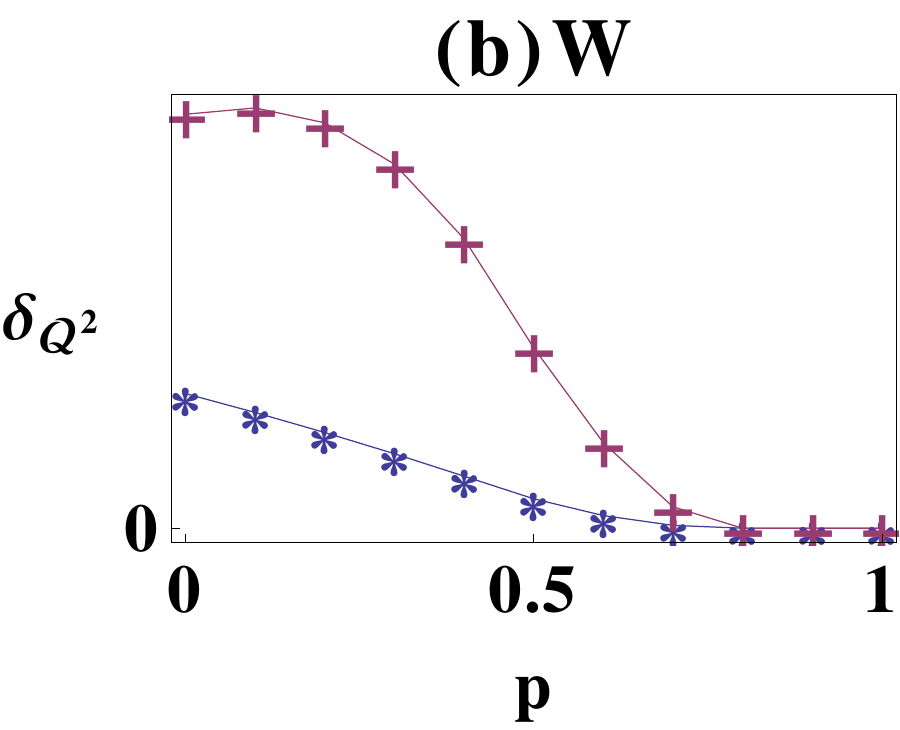}
\caption{(Color online.) Plots of monogamy scores, $\delta_{{\cal Q}^r}(\rho_{ABC})={\cal Q}^r(\rho_{A(BC)})-{\cal Q}^r(\rho_{AB})-{\cal Q}^r(\rho_{AC})$, of negativity (stars) and logarithmic-negativity (pluses) against noise parameter $p$ of W state mixed with white noise, $\rho_{ABC}=(1-p)|W\rangle \langle W| + p \frac{\mathbb{I}}{8}$. For demonstration, power of the quantum correlation measures that we have considered is (a) $r=1$, and (b) $r=2$. For W state, unlike GHZ state, while negativity is monogamous for $r \geq 1$, logarithmic-negativity is monogamous for $r \geq 2$. While $x$-axis is dimensionless, $y$-axis is in ebits.}
\label{fig:w}
\end{figure}

\begin{remark}
A particularly interesting scenario is the following. Suppose that a quantum correlation measure ${\cal Q}$ is monogamous for its $r^{th}$-power. It is important to know the least power, $r^*$, for which the monogamy relation of ${\cal Q}$ is preserved. That is, for what power does a monogamous measure become non-monogamous, and vice-versa? This situation is extremely demanding for generic quantum correlation measures and generic quantum states. 
Moreover, if quantum measure ${\cal Q}$ is monogamous for $r \geq 1$ power, then it will become non-monogamous for $\alpha \leq 0$. That is, if ${\cal Q}^r(\rho_{AB}) \geq \sum_{j=1}^n{\cal Q}^r(\rho_{AB_j})$, then ${\cal Q}^{\alpha}(\rho_{AB}) \leq \sum_{j=1}^n{\cal Q}^{\alpha}(\rho_{AB_j})$ for $\alpha \leq 0$. 
As specific examples, we give plots of monogamy scores, $\delta_{{\cal Q}^r}(\rho_{ABC})={\cal Q}^r(\rho_{A(BC)})-{\cal Q}^r(\rho_{AB})-{\cal Q}^r(\rho_{AC})$, of negativity \cite{negativity} and logarithmic-negativity \cite{logneg} against noise parameter $p$, of GHZ state, $|GHZ\rangle = \frac{1}{\sqrt{2}} (|000\rangle + |111\rangle)$, and W state, $|W\rangle = \frac{1}{\sqrt{3}} (|001\rangle + |010\rangle + |100\rangle)$, mixed with white noise in Figs. \ref{fig:ghz} and \ref{fig:w}. Power of negativity and logarithmic-negativity that we have considered for illustartion is (a) $r=1$, and (b) $r=2$.
We see that for GHZ state, both negativity and logarithmic negativity are monogamous for $r \geq 1$ (see Fig. \ref{fig:ghz}). On the other hand for W state, while negativity is monogamous for $r \geq 1$, logarithmic negativity is monogamous for $r \geq 2$ (see Fig. \ref{fig:w}).
From Fig. \ref{fig:w}(a), we see that logarithmic-negativity is non-monogamous for W state ($p=0$) when $r=1$. However, from Fig. \ref{fig:lneg-w}, we find that it remains non-monogamous upto $r^* \approx 1.06$ (upto second decimal point), and becomes monogamous when $r \gtrsim 1.06$.  
\end{remark}
\begin{figure}[htb]
\includegraphics[width=6cm]{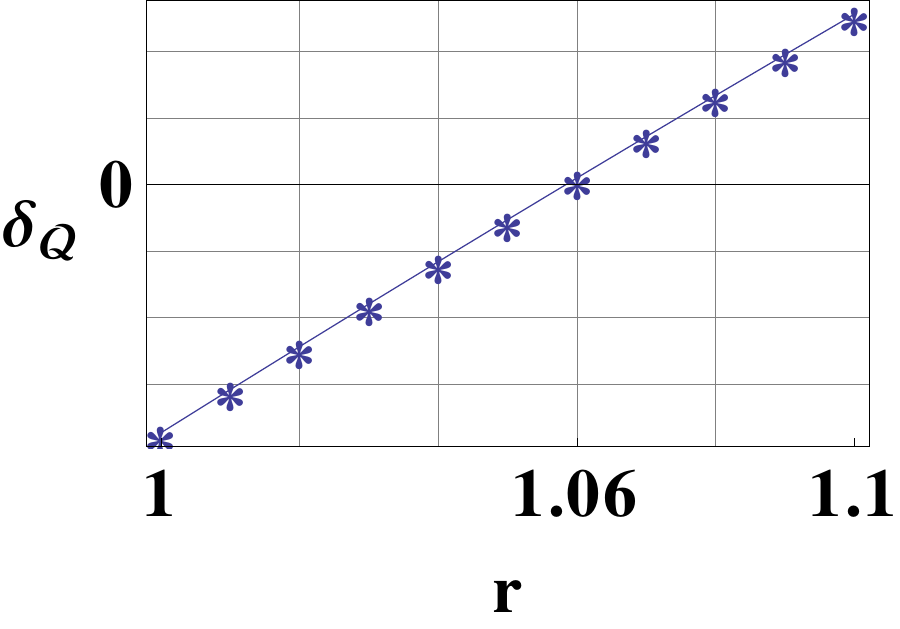}
\caption{(Color online.) 
Plot showing the transition from non-monogamy to monogamy of logarithmic-negativity for W state. Logarithmic-negativity is non-monogamous for W state at $r=1$. It remains non-monogamous upto $r \approx 1.06$ (upto second decimal point), beyond which it becomes monogamous. That is, the minimum power, $r^*$, for which logarithmic-negativity becomes monogamous for W state is $1.06$. While $x$-axis is dimensionless, $y$-axis is in ebits.}
\label{fig:lneg-w}
\end{figure}

Sometimes an entanglement measure ${\cal Q}$ can be a function of another entanglement measure $q$, say, ${\cal Q}=f(q^r)$. Depending on the nature of function $f$ and monogamy of $q$, the monogamy properties of ${\cal Q}$ can be derived. 
For instance, in the seminal paper of CKW \cite{monogamy}, it was already pointed out that any monotonic convex function
of squared concurrence would also be a monogamous measure of entanglement. We extend this observation for a general quantum correlation measure in the following theorem.
\begin{theorem}
\label{thm-functional1}
For an arbitrary multipartite quantum state $\rho_{AB}$, given that $q^r(\rho_{AB}) \geq \sum_j q^r(\rho_{AB_j})$ and ${\cal Q}=f(q^r)$, where $f$ is a monotonically increasing convex function for which $f^m \left(\sum_jq^r_j\right) \geq \sum_jf^m(q^r_j)$, we have ${\cal Q}^m(\rho_{AB}) \geq \sum_j {\cal Q}^m(\rho_{AB_j})$, where $r$ and $m$ are some positive numbers.
\end{theorem}

\begin{proof}
Let $\rho_{AB}=\sum_i p_i |\psi^i\rangle_{AB}\langle \psi^i|$ be the optimal decomposition of $\rho_{AB}$ for ${\cal Q}$. Then
\begin{align}
{\cal Q}(\rho_{AB}) &=\sum_i p_i {\cal Q}(|\psi\rangle^i_{AB}) \nonumber \\
&=\sum_i p_i f(q^r_{ABi}) \nonumber \\
&\geq f\left(\sum_i p_i q^r_{ABi}\right) \nonumber \\
&\geq f\left(\left(\sum_i p_i q_{ABi}\right)^r\right) \nonumber \\
&\geq f\left(q^r(\rho_{AB})\right) \nonumber \\
&\geq f\left(\sum_j q^r_{AB_j}\right) \nonumber \\
&\geq \left(\sum_j f^m(q^r_{AB_j})\right)^{\frac{1}{m}} \nonumber \\
&= \left(\sum_j {\cal Q}^m_{AB_j}\right)^{\frac{1}{m}},
\end{align}
where the first inequality is due to convexity of $f$, the second is due to monotonically increasing nature of $f$ and $\sum_i p_i x_i^r \geq \left(\sum_i p_i x_i\right)^r$, the third is due to monotonicity of $f$ and because $\rho_{AB}=\sum_i p_i |\psi^i\rangle_{AB}\langle \psi^i|$ may not be the optimal decomposition of $\rho_{AB}$ for $q$ (that is, $q(\rho_{AB}) \leq \sum_i p_i q(|\psi\rangle^i_{AB})$), the fourth is due to monogamy of $q^r$, and the fifth inequality follows from the constraint $f^m(\sum_jq^r_j) \geq \sum_jf^m(q^r_j)$.
Hence the theorem is proved.
\end{proof} 

The monogamy of squared EoF can be stated as a corollary of Theorem \ref{thm-functional1}.
\begin{corollary}
\label{cor-eof2}
The square of entanglement of formation is monogamous. 
\end{corollary}

\begin{proof}
EoF is a concave function of squared concurrence given by
\begin{equation}
\label{eq:eof}
{\cal E}(\rho_{AB})=h\left(\frac{1+\sqrt{1-\mathcal{C}^2(\rho_{AB})}}{2}\right),
\end{equation}
where $h(x)=-x\log_2x-(1-x)\log_2(1-x)$ is the Shannon (binary) entropy. 
However, squared EoF is a convex monotonic function of squared concurrence and satisfies 
${\cal E}^2(\sum_j \rho_{AB_j}) \geq \sum_j {\cal E}^2(\rho_{AB_j})$. The corollary follows because squared concurrence is monogamous \cite{monogamy, osborne}.
\end{proof}
Independent proofs of monogamy of squared EoF have been provided recently in Refs. \cite{mono-eof2-bai, mono-eof2-fanchini}.\\

\begin{remark}
\label{thm-functional2}
Using the same line of proof as in Theorem \ref{thm-functional1}, we can show that 
for an arbitrary multipartite quantum state $\rho_{AB}$, given that $q^r(\rho_{AB}) \geq \sum_j q^r(\rho_{AB_j})$ and ${\cal Q}=f(q^r)$, where $f$ be a monotonically decreasing concave function for which $f^m \left(\sum_jq^r_j\right) \geq \sum_jf^m(q^r_j)$, we have ${\cal Q}^m(\rho_{AB}) \leq \sum_j {\cal Q}^m(\rho_{AB_j})$, where $r$ and $m$ are some positive numbers.\\
\end{remark}

Next we asked whether there is any correspondence between monogamy of a quantum correlation measure for arbitrary pure and that of mixed states. This led us to the result in Theorem \ref{thm-mixed}, and the remarks and corollary following it.
\begin{theorem}
\label{thm-mixed}
If a convex bipartite quantum correlation measure ${\cal Q}$ when raised to power \(r=1,2\) is monogamous for pure multipartite states, then ${\cal Q}^r$ is also monogamous for the mixed states in the given Hilbert space.
\end{theorem}

\begin{proof}
Convexity of ${\cal Q}$ implies that if $\rho=\sum_i p_i \rho^i$ then ${\cal Q}(\rho) \leq \sum_i p_i {\cal Q}(\rho^i)$.
Assume that ${\cal Q}^r$, (\(r=1,2\)), is monogamous for arbitrary multipartite pure state $|\psi\rangle_{AB}=|\psi\rangle_{AB_1B_2 \cdots B_n}$ in some Hilbert space of dimension $d_A \otimes d_{B_1} \otimes d_{B_2} \cdots \otimes d_{B_n}$. That is,
\begin{equation}
{\cal Q}^r(|\psi\rangle_{AB}) \geq \sum_{j=1}^n{\cal Q}^r(\rho^{\psi}_{AB_j}).
\end{equation}
Let $\rho_{AB}=\sum_i p_i |\psi^i\rangle_{AB}\langle \psi^i|=\sum_i p_i \rho^i_{AB}$ be the optimal decomposition of $\rho_{AB}$ for ${\cal Q}$, and  $\rho^i_{AB_j}=\mbox{tr}_{\overline{AB_j}} \rho^i_{AB}$,  $\rho_{AB_j}=\mbox{tr}_{\overline{AB_j}} \rho_{AB}$ be the reduced density matrices obtained after partial-tracing the sub-systems except $A \ \ and \ \ B_j$ ($j=1,2,\cdots,n$).\\

When \(r=1\), we have 
${\cal Q}(\rho_{AB})=\sum_i p_i{\cal Q}(|\psi^i\rangle_{AB})
 \geq \sum_i p_i \sum_j {\cal Q}(\rho^i_{AB_j})
 = \sum_j \left(\sum_i p_i {\cal Q}(\rho^i_{AB_j})\right) 
 \geq \sum_j {\cal Q}\left(\sum_i p_i \rho^i_{AB_j}\right)
 = \sum_j {\cal Q}(\rho_{AB_j})$,
where the first inequality is due to monogamy of ${\cal Q}$ for pure states and the second inequality is due to the convexity of ${\cal Q}$.\\

When \(r=2\), let us write
\begin{align}
{\cal Q}(\rho_{AB})&=\sum_i p_i{\cal Q}(|\psi^i\rangle_{AB})=\sum_i {\cal Q}_{ABi} \\
{\cal Q'}(\rho_{AB_j})&=\sum_i p_i{\cal Q}(\rho^i_{AB_j})=\sum_i {\cal Q}_{AB_ji} \nonumber \\
& \geq {\cal Q}(\rho_{AB_j})
\label{convexity-Q}
\end{align}
The above inequality follows from convexity of ${\cal Q}$. We, then, have the following inequality
\begin{eqnarray}
&&{\cal Q}^2(\rho_{AB})-\sum_j {\cal Q'}^2(\rho_{AB_j}) \nonumber \\
&=&\left(\sum_i {\cal Q}_{ABi}\right)^2-\sum_j \left(\sum_i {\cal Q}_{AB_ji}\right)^2 \nonumber \\
&=&\sum_i \left({\cal Q}^2_{ABi}-\sum_j {\cal Q}^2_{AB_ji}\right) \nonumber \\
&&+2\sum_{i=1}^{n-1} \sum_{k=i+1}^n \left({\cal Q}_{ABi}{\cal Q}_{ABk}-\sum_j{\cal Q}_{AB_ji}{\cal Q}_{AB_jk}\right) \nonumber \\
& \geq 0,
\end{eqnarray}
because, in the second equation, the first term is non-negative due to monogamy of ${\cal Q}^2$ for pure states and the second term is non-negative as shown below. We have, for arbitrary pure states 
$|\psi^i\rangle_{AB}$ and $|\psi^k\rangle_{AB}$,
\begin{align}
{\cal Q}^2_{ABi}{\cal Q}^2_{ABk}& \geq \left(\sum_j {\cal Q}^2_{AB_ji}\right) \left(\sum_j {\cal Q}^2_{AB_jk}\right) \nonumber \\
& \geq \left(\sum_j {\cal Q}_{AB_ji}{\cal Q}_{AB_jk}\right)^2,
\end{align} 
where the first inequality is due to monogamy of ${\cal Q}^2$ for pure states while the second inequality follows from the Cauchy-Schwarz inequality, $\sum_i a_ib_i \leq \sqrt{\left(\sum_i a^2_i\right) \left(\sum_i b^2_i\right)}$. Hence,
\begin{equation}
{\cal Q}_{ABi}{\cal Q}_{ABk}-\sum_j{\cal Q}_{AB_ji}{\cal Q}_{AB_jk} \geq 0.
\end{equation}
Since ${\cal Q'}(\rho_{AB_j}) \geq {\cal Q}(\rho_{AB_j})$ (due to convexity of ${\cal Q}$ as shown in Eq. (\ref{convexity-Q})), we obtain the desired monogamy relation for mixed state,
\begin{equation}
{\cal Q}^2(\rho_{AB}) \geq \sum_j {\cal Q}^2(\rho_{AB_j}).
\end{equation}
\end{proof}

\begin{remark}
For \(r > 2\), the monogamy of ${\cal Q}^r$ is as yet inclusive for mixed states, even though monogamy holds for pure states (see Appendix). 
\end{remark}

\begin{remark}
In Ref. \cite{asu-multi}, it was shown numerically that entanglement measures become monogamous for pure states with increasing number of qubits. 
It was also figured out that ``good'' entanglement measures \cite{good-measures} like relative entropy of entanglement, regularized
relative entropy of entanglement \cite{reg-rel-ent}, entanglement cost \cite{ent-cost1, ent-cost2}, distillable entanglement, all of which are not generally computable, are monogamous for almost all pure states of four or more qubits.
Theorem \ref{thm-mixed} then implies that such ``good'' convex entanglement measures will become monogamous for multiqubit mixed states also.
\end{remark}

\begin{corollary}
\label{cor-neg2}
The squared negativity is monogamous for \(n\)-qubit mixed states.
\end{corollary}

\begin{proof}
Negativity is a convex function \cite{negativity}, and it has been  proven that the square of negativity is monogamous for \(n\)-qubit pure states \cite{mono-neg2-fan}. Hence the proof.
\end{proof}

Further, we wanted to explore if we could obtain general and tighter monogamy relations other than the standard one in Eq. (\ref{eq:std-mono-ineq}). This led us to the results in Theorem \ref{thm-hierarchy1} and Theorem \ref{thm-tight1}, and the remarks following the same.
\begin{theorem}
\label{thm-hierarchy1}
{\bf (Hierarchical monogamy relations)}
${\cal Q}^r(\rho_{AXY}) \geq {\cal Q}^r(\rho_{AX})+{\cal Q}^r(\rho_{AY})$ for an arbitrary state $\rho_{AXY}$ implies ${\cal Q}^{\alpha}(\rho_{ABC}) \geq {\cal Q}^{\alpha}(\rho_{AB})+\sum_{j=1}^{k-1} {\cal Q}^{\alpha}(\rho_{AC_j})+ {\cal Q}^{\alpha}(\rho_{AC_k \cdots C_m})$ for $\rho_{ABC} \equiv \rho_{ABC_1\cdots C_m}$ when $\alpha \geq r$.
\end{theorem}

\begin{proof}
First, we will prove the hierarchical monogamy relations using the given condition, ${\cal Q}^r(\rho_{AXY}) \geq {\cal Q}^r(\rho_{AX})+{\cal Q}^r(\rho_{AY})$ for arbitrary state $\rho_{AXY}$, and thereafter we will show that these relations are also valid for $\alpha \geq r$. For multiparty state $\rho_{ABC} \equiv \rho_{ABC_1\cdots C_m}$, applying the given condition repeatedly, we obtain a family of hierarchical monogamy relations as given below
\begin{align}
\label{mono-ABC0}
{\cal Q}^r(\rho_{ABC}) &\geq {\cal Q}^r(\rho_{AB})+{\cal Q}^r(\rho_{AC}) \nonumber \\
& \geq {\cal Q}^r(\rho_{AB})+{\cal Q}^r(\rho_{AC_1})+{\cal Q}^r(\rho_{AC_2 \cdots C_m}) \nonumber \\
& \hspace{1.5cm} \vdots \nonumber \\
& \geq {\cal Q}^{r}(\rho_{AB})+\sum_{j=1}^{k-1} {\cal Q}^{r}(\rho_{AC_j})+ {\cal Q}^{r}(\rho_{AC_k \cdots C_m}) \nonumber \\
& \hspace{1.5cm} \vdots \nonumber \\
&\geq {\cal Q}^{r}(\rho_{AB})+\sum_{j=1}^{m} {\cal Q}^{r}(\rho_{AC_j}).
\end{align}
Now, we will show that these hierarchical monogamy relations are also valid for $\alpha \geq r$.
From Theorem \ref{thm-raising}, ${\cal Q}^r(\rho_{ABC}) \geq {\cal Q}^r(\rho_{AB})+{\cal Q}^r(\rho_{AC})$ implies that
\begin{equation}
{\cal Q}^{\alpha}(\rho_{ABC}) \geq {\cal Q}^{\alpha}(\rho_{AB})+{\cal Q}^{\alpha}(\rho_{AC}).
\label{mono-ABC}
\end{equation}
Now,
\begin{align}
{\cal Q}^{\alpha}(\rho_{AC})&= \left\{{\cal Q}^{r}(\rho_{AC_1(C_2 \cdots C_m)})\right\}^{\frac{\alpha}{r}} \nonumber \\
&\geq \left\{{\cal Q}^{r}(\rho_{AC_1})+{\cal Q}^{r}(\rho_{A(C_2 \cdots C_m)})\right\}^{\frac{\alpha}{r}} \nonumber \\ 
&\geq \sum_{j=1}^{k-1} {\cal Q}^{\alpha}(\rho_{AC_j})+{\cal Q}^{\alpha}(\rho_{AC_k \cdots C_m}) \\
&\geq \sum_{j=1}^{m} {\cal Q}^{\alpha}(\rho_{AC_j}).
\end{align}
Thus, we obtain inequalities
\begin{equation}
{\cal Q}^{\alpha}(\rho_{ABC}) \geq {\cal Q}^{\alpha}(\rho_{AB})+\sum_{j=1}^{k-1} {\cal Q}^{\alpha}(\rho_{AC_j})+{\cal Q}^{\alpha}(\rho_{AC_k \cdots C_m}), 
\label{result1}
\end{equation}
and
\begin{equation}
{\cal Q}^{\alpha}(\rho_{ABC}) \geq {\cal Q}^{\alpha}(\rho_{AB})+\sum_{k=1}^m {\cal Q}^{\alpha}(\rho_{AC_k}).
\label{result2}
\end{equation}
\end{proof}

\begin{remark}
\label{thm-hierarchy2}
Using Theorem \ref{thm-lowering} and the same line of proof as in Theorem \ref{thm-hierarchy1}, one can prove {\em hierarchical non-monogamy relations}.
${\cal Q}^r(\rho_{XYZ}) \leq {\cal Q}^r(\rho_{XY})+{\cal Q}^r(\rho_{XZ})$ for an arbitrary state $\rho_{XYZ}$ implies ${\cal Q}^{\alpha}(\rho_{ABC}) \leq {\cal Q}^{\alpha}(\rho_{AB})+\sum_{j=1}^{k-1} {\cal Q}^{\alpha}(\rho_{AC_j})+ {\cal Q}^{\alpha}(\rho_{AC_k \cdots C_m})$ for $\rho_{ABC} \equiv \rho_{ABC_1\cdots C_m}$ when $\alpha \leq r$.
\end{remark}

\begin{theorem}
\label{thm-tight1}
{\bf (Strong monogamy inequality)}
If ${\cal Q}^r(\rho_{AB}) \geq \sum_j {\cal Q}^r(\rho_{AB_j})$ for an arbitrary multipartite quantum state $\rho_{AB_1\cdots B_n} \equiv \rho_{AB}$, then ${\cal Q}^{\alpha}(\rho_{AB}) \geq \frac{1}{2^{n-1}-1}\sum_X {\cal Q}^{\alpha}(\rho_{AX}) \geq \sum_j {\cal Q}^{\alpha}(\rho_{AB_j})$ for $\alpha \geq r \geq 1$, where \(X\) is the composite system corresponding to some nonempty proper subset of $B = \{B_1,B_2,\cdots,B_n \}$.
\end{theorem}

\begin{proof}
Here again we can split the proof in two parts, as in the proof of Theorem \ref{thm-hierarchy1}. 
For instance, first we can obtain the monogamy relation, ${\cal Q}(\rho_{AB}) \geq \frac{1}{2^{n-1}-1}\sum_{X} {\cal Q}\left(\rho_{AX}\right) \geq \sum_{j} {\cal Q}\left(\rho_{AB_j}\right)$, and then show that such a monogamy relation is also true for the $\alpha$-th power. However, for the sake of brevity, we will start with the $\alpha$-th power.
Let $B = \{B_1,B_2,\cdots,B_n \}$ be the set of
subsystems $B_i$'s, and $X=\{B_{i_1},\cdots,B_{i_k}\}$ and $X^c=B-X$ be nonempty proper subsets of $B$. Thus $\rho_{AB}=\rho_{AXX^c}$. Applying monogamy inequality and Theorem \ref{thm-raising}, we get
\begin{align}
{\cal Q}^{\alpha}\left(\rho_{AB}\right) \geq {\cal Q}^{\alpha}\left(\rho_{AX}\right)+{\cal Q}^{\alpha}\left(\rho_{AX^c}\right).
\label{monoAXXC}
\end{align}
Since the set of all nonempty proper subsets of $B$ is same as the set of their complements, i.e.,
$\{X | X \subset B \}=\{X^c| X \subset B \}$, summing over all possible nonempty proper subsets $X$'s of $B$ leads to the following inequality,
\begin{align}
{\cal Q}^{\alpha}\left(\rho_{AB}\right) &\geq \frac{1}{2^n-2}\sum_{X}\left({\cal Q}^{\alpha}\left(\rho_{AX}\right)+{\cal Q}^{\alpha}\left(\rho_{AX^c}\right)\right) \nonumber\\
&= \frac{1}{2^{n-1}-1}\sum_X {\cal Q}^{\alpha}(\rho_{AX}).
\label{mono-sets}
\end{align}
We also have 
\begin{align}
{\cal Q}^{\alpha}\left(\rho_{AX}\right)+{\cal Q}^{\alpha}\left(\rho_{AX^c}\right) \geq \sum_j {\cal Q}^{\alpha}\left(\rho_{AB_j}\right).
\label{upper1}
\end{align}
Again summing over all possible nonempty proper subsets $X$'s of $B$, we obtain
\begin{align}
\frac{1}{2^{n-1}-1}\sum_{X} {\cal Q}^{\alpha}\left(\rho_{AX}\right) \geq \sum_{j} {\cal Q}^{\alpha}\left(\rho_{AB_j}\right).
\label{upper2}
\end{align}
Combining inequalities~(\ref{mono-sets}) and (\ref{upper2}), we obtain the desired {\em strong monogamy inequality} for arbitrary multi-party quantum state $\rho_{AB_1B_2\cdots B_n}$.
\end{proof}


\begin{remark}
\label{thm-tight2}
It was shown in Ref. \cite{strong-poly-kim} that entanglement of assistance \cite{eoa} follows strong non-monogamy relation.
Using Theorem \ref{thm-lowering} and the same line of proof as in Theorem \ref{thm-tight1}, we can prove that if ${\cal Q}^r(\rho_{AB}) \leq \sum_j {\cal Q}^r(\rho_{AB_j})$ for any multipartite quantum state $\rho_{AB_1\cdots B_n} \equiv \rho_{AB}$, then ${\cal Q}^{\alpha}(\rho_{AB}) \leq \frac{1}{2^{n-1}-1}\sum_X {\cal Q}^{\alpha}(\rho_{AX}) \leq \sum_j {\cal Q}^{\alpha}(\rho_{AB_j})$ for $\alpha \leq r$.
\end{remark}

\section{conclusion}
\label{sec:summary}
Monogamy is one of the most important properties for many-body quantum systems, which restricts sharing of quantum correlations among many parties and there is a trade-off among the amounts of quantum correlations in different subsystems and partitions. It is also a distinguishing feature between quantum and classical correlations. Moreover, it has played a significant role in devising quantum security in secret key generation and multiparty communication protocols, besides being a useful tool in exploring nonclassical correlations in multiparty systems.
In this letter, we have explored the conditions under which monogamy of functions of quantum correlation measures is preserved. 
We have shown that a monogamous measure remains monogamous on raising its power, and a non-monogamous measure remains non-monogamous on lowering its power. We have also proven that monogamy of a convex quantum correlation measure for arbitrary multipartite pure quantum states leads to its monogamy for the mixed states. 
This significantly simplifies the task of establishing the monogamy relations for mixed states.
Our study partially answers the two conjectures in Ref. \cite{hier-eof2} in the sense that it now only remains to prove the monogamy of the squared entanglement of formation for pure states in arbitrary dimensions.
Monogamy of squared negativity for mixed states and that of squared entanglement of formation turn out to be special cases of our results.
Furthermore, we have examined hierarchical monogamy relations and tighter monogamy inequalities compared to the standard one. 

\begin{acknowledgements}
AK is very thankful to Ujjwal Sen and Mallesham K. for useful discussions and reading the manuscript. 
\end{acknowledgements}

\begin{center}
\textbf{Appendix}\\~\\
\end{center}

\textit{Theorem \ref{thm-mixed} cannot be stated conclusively for \(r > 2\) by using the same line of proof as in the main text.}\\

Multinomial expansion is given by
\begin{equation}
\left(\sum_i x_i\right)^r=\sum_{\{r_k\}} \frac{r!}{\prod_k r_k!} \prod_i x_i^{r_k}
\end{equation}
where $\{r_k | 0 \leq r_k \leq r \ \ \& \sum_k r_k=r \}$ is the integer partition of \(r\), and the summation is over all permutations of such integer partitions of \(r\). Then, as in Theorem \ref{thm-mixed}, we have 
\begin{eqnarray}
&&{\cal Q}^r(\rho_{AB})-\sum_j {\cal Q'}^r(\rho_{AB_j}) \nonumber \\
&=&\left(\sum_i {\cal Q}_{ABi}\right)^r-\sum_j \left(\sum_i {\cal Q}_{AB_ji}\right)^r \nonumber \\
&=&\sum_{\{r_k\}} \frac{r!}{\prod_k r_k!} \left(\prod_i {\cal Q}_{ABi}^{r_k}-\sum_j \prod_i {\cal Q}_{AB_ji}^{r_k}\right) \nonumber \\
&=&\sum_i\left({\cal Q}^r_{ABi}-\sum_j {\cal Q}^r_{AB_ji}\right) \nonumber \\
&&+\sum_{\{r_k \neq r\}} \frac{r!}{\prod_k r_k!} \left(\prod_i {\cal Q}_{ABi}^{r_k}-\sum_j \prod_i {\cal Q}_{AB_ji}^{r_k}\right).
\end{eqnarray}
Although, in the third equality, the first term is non-negative due to the monogamy of ${\cal Q}^r$ for pure states, we cannot say anything with certainty about the second term as we do not have {\em Holder-type inequality for multi-variables}. However, the other way is always true, i.e., if ${\cal Q}^r$ is monogamous for mixed states then it is certainly monogamous for pure states.
\hfill $\blacksquare$
\end{document}